\documentclass[12pt]{article}
\usepackage{Tillpack}

\title{Trees in simple Polygons}
\author{ Tillmann Miltzow\footnote{Institute of Computer Science, Freie Universit\"at Berlin, Germany. \texttt{t.miltzow@gmail.com}}} 

\date{}

\begin{document}
\maketitle

\begin{abstract} 
	We prove that every simple polygon contains a degree $3$ tree encompassing a prescribed set of vertices.
	We give tight bounds on the minimal number of degree $3$ vertices. We apply this result to reprove a result from %\citet{Bose_everyset}
	Bose~\emph{et al.}~\cite{DBLP:journals/dcg/BoseHT01} that every set of disjoint line segments in the plane admits a binary tree.
\end{abstract}

\paragraph{Introduction}
Recently many papers have been published regarding the augmentation of discrete geometric objects, in particular a set of non-crossing segments in the plane. While these problems are studied, not just many open problems could be solved but also new tools developed. Among these tools is the simple and easy to prove lemma \emph{about matchings in polygons} by Manuel Abellanas, Alfredo Garc\'{\i}a Olaverri, Ferran Hurtado, Javier Tejel and Jorge Urrutia~\cite{DBLP:journals/comgeo/AbellanasOHTU08}. It gives sufficient conditions when a predefined set of vertices of a polygon can be \emph{geometrically} matched \emph{within} this polygon. 
It was recently used to show that the bichromatic compatible matching graph is connected~\cite{2012arXiv1207.2375A}.

The main motivation of this note is to find a useful variation of this auspicious lemma about matchings in polygons. Trees with low maximal degree seemed curious. We show the existence of a max degree $3$ tree spanning a prescribed set of vertices \emph{in} a given polygon(to be made precise soon). We give explicit tight bounds on the number of degree $3$ vertices in the attained tree. To the best of our knowledge this question has not been studied before.

As an application we give a new easy proof to an old result from Prosenjit Bose, Michael E. Houle and Godfried T. Toussaint~\cite{DBLP:journals/dcg/BoseHT01} that every set of non-crossing line segments in the plane admits a binary encompassing geometric tree.
Later Michael Hoffmann, Bettina Speckmann and Csaba D. T\'{o}th gave an alternative proof~\cite{Hoffmann201035}. A comparison of the their proof and ours reveals similarities, which will be explained later. 

It is important to note that we assume general position throughout this note in the sense that no three segment endpoints are on a common line. Theorem \ref{thm:main} also holds without this assumption.

\paragraph{Main Result}
A PSLG (planar straight line graph) is a planar non-crossing embedding of a graph with straight edges. 
A simple polygon is a PSLG which is a cycle. We say that a PSLG $G$ is \emph{inscribed in} a simple polygon $P$ if  $G$ is contained in the closed bounded region of $P$ and for every edge $e$ of $G$ either: $e$ is in the interior of $P$ or
 $e$ is also an edge of $P$.

We say a vertex of a simple polygon is \emph{reflex} if the interior angle to its incident edges is larger than $\pi$ and \emph{convex} otherwise. 
	\begin{figure}[h]
			\begin{center}
			\includegraphics[width = 0.4\textwidth]{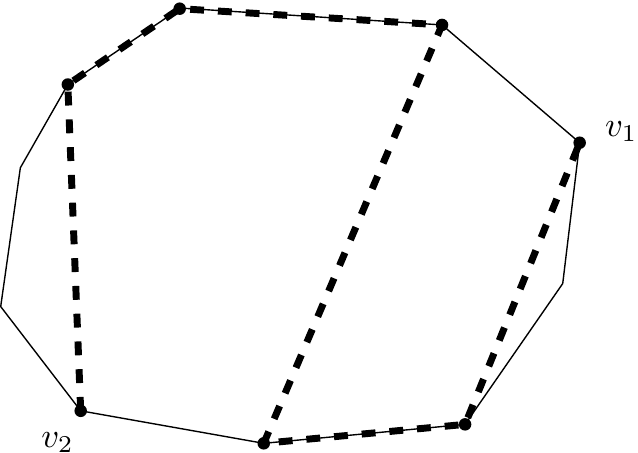}
				\caption{The Theorem is trivial in the convex case.}
				\label{fig:convexCase}
			\end{center}
	\end{figure} 
\begin{theorem}\label{thm:main}
	Let $P$ be a simple polygon, $R$ the set of reflex vertices of $P$ and $A\supseteq R$ some superset of vertices of $P$. Further we can specify two vertices $v_1,v_2 \in A\backslash R$ such that the following holds:
	There exists a PSLG $T$ with vertex set $A$ inscribed in $P$, which is a tree. The tree $T$ satisfies the following degree restrictions:
$$
\begin{array}{rl}
	 deg(v) \leq  3 & \ \ \forall \ v\in R  \\
	 deg(v)  \leq  2 & \ \ \forall \ v\in A\backslash R  \\
	 deg(v)  =  1 & \ \ \forall \ v\in \{ v_1, v_2 \}.  \\
	\end{array}
	$$
	For every natural number $n$ exists a polygon with $n$ reflex vertices
	such that any tree with the properties above has $n$ reflex vertices.
\end{theorem}
\begin{proof}
	The proof goes by induction on the number of reflex vertices of $P$.
	For the induction basis assume $P$ is convex. 
	We aim to construct a path from $v_1$ to $v_2$ traversing
	the vertices of $A$. Consider the only two paths $Q_1$ and $Q_2$ on $P$ connecting 
	$v_1$ and $v_2$. One possible path from 
	$v_1$ to $v_2$ traverses at first all vertices of $A$ on $Q_1$ in the order 
	they appear on $Q_1$ and traverses thereafter the remaining vertices on $Q_2$,
	before it reaches $v_2$.
	Clearly, $deg(v_1) = deg(v_2) = 1$ and all other vertices have degree $2$, see Figure \ref{fig:convexCase}. 	It is an instructive exercise to count the number of such paths in a given convex polygon.
	\begin{figure}[h]
			\begin{center}
			\includegraphics[width = \textwidth]{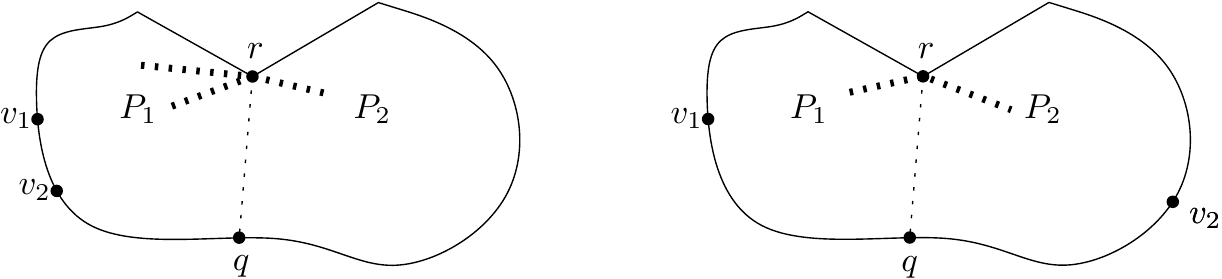}
				\caption{Either both marked vertices are in $P_1$ or in $P_2$.}
				\label{fig:inductionStep}
			\end{center}
	\end{figure} 
	
	For the induction step assume that $P$ has some reflex vertices and let $r$ be such a reflex vertex of $P$, as in Figure \ref{fig:inductionStep}. We shoot a ray from $r$ into the interior of $P$ in such a way that both angles at $r$ become convex. We denote by $q$ the first point on the boundary of $P$  hit by the ray. We further assume that $q$ is not a vertex of $P$. The segment $rq$ splits $P$ into two simple polygons $P_1$ and $P_2$. The vertices $r$ and $q$ are convex in $P_1$ and $P_2$. Both $P_1$ and $P_2$ have fewer reflex vertices than $P$. We denote the vertices of $A$ in $P_i$ with $A_i$. We list the vertex $r$ in $A_1$ and $A_2$. We distinguish two cases. 
	
	\textbf{Case 1} Both $v_1$ and $v_2$ specified in the Lemma are in the same polygon, say $P_1$. Then we apply the induction hypothesis on $P_1$ with the $A_1$ and $v_1$ and $v_2$ as specified vertices. We receive a PSLG $T_1$ with the above properties. In particular $r$ has degree $\leq 2$ in $T_1$. We also apply the induction hypothesis on $P_2$ and $A_2$ with $r$ and any other vertex of $A_2\backslash R$ specified. We receive a second tree $T_2$ with the above properties and $r$ has degree $1$ in $P_2$. The union $T = T_1 \cup T_2$ has all the above properties and in particular $r$ has degree at most $3$.
	
	\textbf{Case 2} The specified vertices $v_1$ and $v_2$ are in different polygons $P_1$ and $P_2$ respectively. We apply the induction hypothesis on $P_1$ with $A_1$ and  $v_1$ and $r$ as specified vertices. We receive a PSLG $T_1$ with all the above properties and the degree of $r$ in $T_1$ is exactly one. Similarly there exists a PSLG $T_2$ in $P_2$ such that $r$ has degree $1$ in $T_2$ the union $T = T_1 \cup T_2$  satisfies all the properties in the Lemma and the degree of $r$ is exactly $2$.
				\begin{figure}[h]
			\begin{center}
			\includegraphics[width = \textwidth]{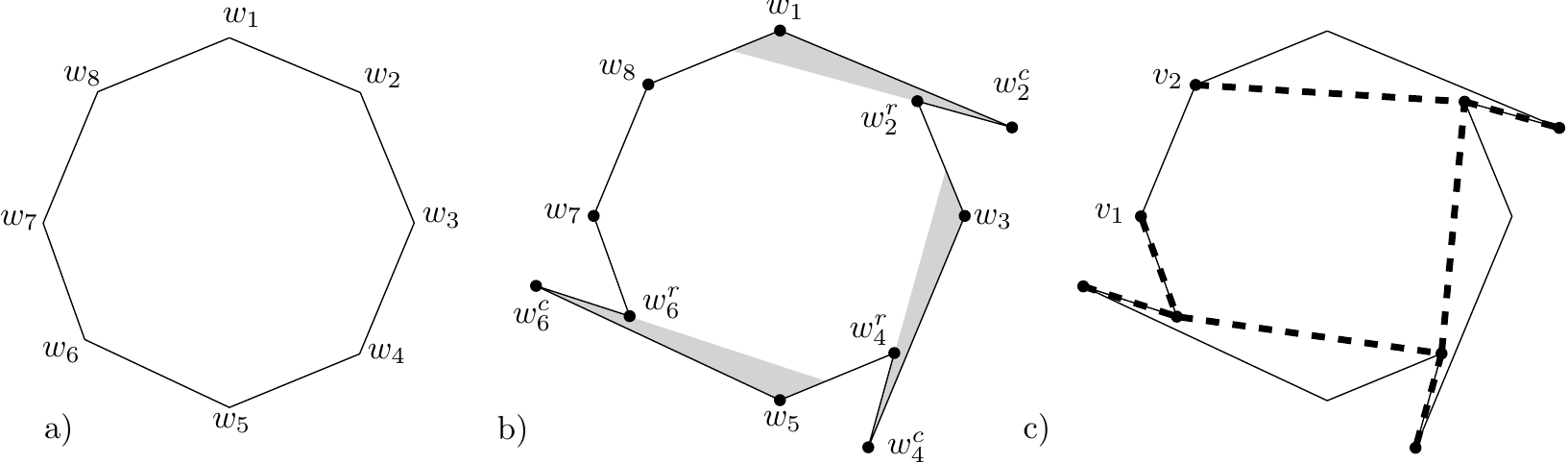}
				\caption{a) Is the drawing of a regular polygon on $8$ vertices. b) Is the transformed vertices with $3$ spikes added. The shaded region indicates the visibility region of the convex vertex of the spikes. c) The only tree that satisfies the conditions of Theorem \ref{thm:main} is drawn with fat dashed lines. All reflex vertices must have degree $3$.}
				\label{fig:tightExample}
			\end{center}
	\end{figure} 
	It remains to give an example of a polygon with $n$ reflex vertices such that at any reflex vertex we must have degree $3$. We start the construction with a regular polygon on $2n+2$ vertices. We add a spike to the vertices $w_2, w_4, w_6, \ldots ,w_{2n}$, as shown in Figure~\ref{fig:tightExample}~b).
	This creates new reflex vertices $w^r_2, w^r_4, \ldots ,w^r_{2n}$ and new convex vertices $w^c_2, w^c_4, \ldots ,w^c_{2n}$ and destroys the corresponding original vertices of the polygon. For any new convex vertex $w_k^c$ only the vertices $w_k^r$ and $w_{k-1}$ are visible, by construction. We define $A$ formally as $$ A = \{ w^r_2,  \ldots ,w^r_{2n} \} \cup \{  w^c_2,  \ldots ,w^c_{2n} \} \cup \{w_{2n+1},w_{2n+2}\}.$$
	 We choose  $w_{2n+1}$ and $w_{2n+2}$ as the marked vertices $v_1$ and $v_2$. It is clear, that in any encompassing PSLG as in Theorem \ref{thm:main} the vertices at the spikes must have degree $1$ as well as $v_1$ and $v_2$. Since the sum of the degrees is constant, exactly the reflex vertices must have degree $3$.
\end{proof}

The just proven Theorem needed the marked vertices $v_1$ and $v_2$ mainly for the induction to work. One could formulate the Theorem without mentioning them and only consider them in the proof, because they strengthen the result only slightly. 
To do so consider the case that $A$ contains at least $2$ convex vertices and mark them as $v_1$ and $v_2$. In the other case consider
vertices with smallest and largest $x$-coordinate.  They clearly are convex. Mark them as $v_1$ and $v_2$. They can be remove later from the tree, as they are leaves.

\paragraph{Application}In this paragraph we want to prove the existence of a binary tree spanning a set of disjoint line segments in the plane.
It was proven first by Bose~\emph{et al.}~\cite{DBLP:journals/dcg/BoseHT01}. Later Hoffmann~\emph{et al.} gave a simpler alternative proof~\cite{Hoffmann201035}. We will give a third proof, as an application of Theorem \ref{thm:main}. Our proof resembles the proof in \cite{Hoffmann201035}. Roughly speaking Hoffmann~\emph{et al.} use a convex subdivision and the \emph{tunnel graph} to assign the segment endpoints to neighboring cells in a clever way. We only construct one polygon, apply Theorem \ref{thm:main} and delete some superfluous edges. Implicitly we also subdivide the plane into convex regions, but our assignment to the regions is more crude. That is why we have an additional clean up step in the end.

A PSLG is an \emph{encompassing tree} of a set $S$ of line segments in the plane if it is a tree as a graph and every line segment of $S$ is an edge of the tree.
\begin{theorem}[\cite{DBLP:journals/dcg/BoseHT01}]
	Let $S$ be a set of $n$ disjoint line segments in the plane with no $3$ segment endpoints on a common line. Then there exists 
a max degree-$3$ planar encompassing tree of $S$. 
\end{theorem}
\begin{proof}
	The idea of the proof is to construct a simple polygon out of the set of line segments, apply Theorem \ref{thm:main} and make some minor changes in order to  attain the desired tree, see Figure~\ref{fig:overview}.
	\begin{figure}[h]
			\begin{center}
			\includegraphics[width = \textwidth]{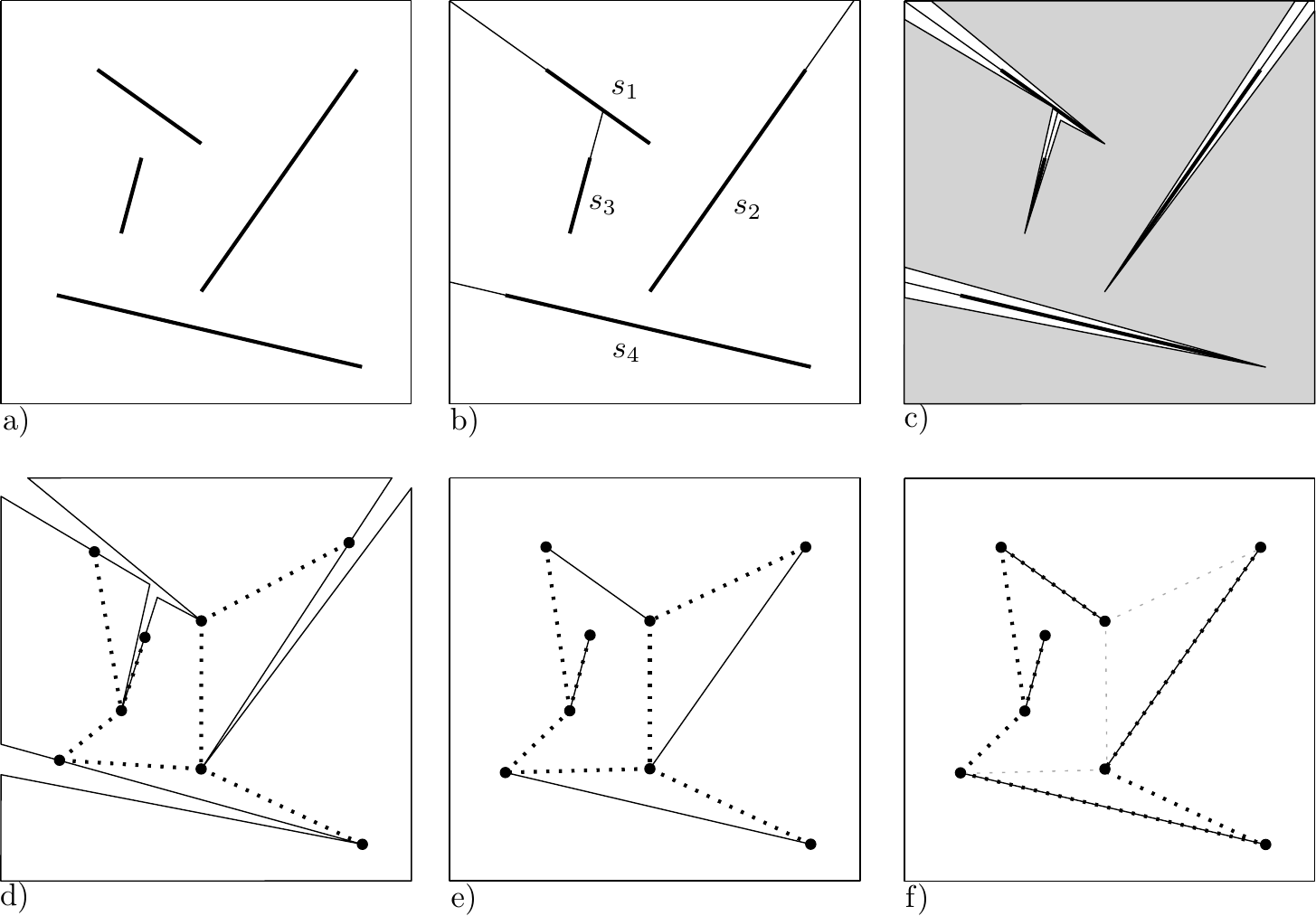}
				\caption{a) Line segments are surrounded by a bounding box. b) The polygon $P_n$ is attained after the segment extensions, note that $s_3$ can only be extended after $s_1$. c) The simplified polygon is colored grey. d) The tree attained by Theorem \ref{thm:main} is depicted. 
				e) The corresponding tree of the line segments.
				 f) Adding the original line segments and removing superfluous edges in grey results in the desired binary tree.}
				\label{fig:overview}
			\end{center}
	\end{figure} 
	The first step is to draw a bounding box $P_0$ around $S$. Next we pick a segment $s_1$ of $S_1 = S$ which has a point on the boundary of the convex hull of $S_1$. We extend $s_1$ till it reaches the bounding box and becomes part of the new bounding polygon $P_1$ itself.
	We continue iteratively. The set $S_i$ is defined to be $S_{i-1} \backslash \{s_{i-1}\}$ and $s_i \in S_i$ is a segment which has a point on the boundary of the convex hull of $S_i$. We extend $s_i$ till it becomes part of the new bounding polygon $P_i$ itself. See Figure~\ref{fig:overview}~b) for an illustration.
Clearly $P_n$ is a polygon, though not simple. In the next step depicted in Figure~\ref{fig:overview}~c), we draw a simple polygon $Q$ in the interior of $P_n$. We want that the set of reflex vertices of $Q$ and $P_n$ agree and that $Q$ is very close to $P_n$ in the Hausdorff sense. As this is very intuitive we hope that the reader will be satisfied to know that
this simplification step was used also for instance in \cite{2012arXiv1207.2375A} lemma $1$ and we refer to them for a rigorous proof.
We want to apply Theorem \ref{thm:main} on $Q$.
Note that each line segment has exactly one reflex vertex in $Q$ and these are the only reflex vertices of $Q$. We define $A$ as the set of reflex vertices of $Q$ and for each extended endpoint of a line segment we choose a closest point on $Q$ to belong to $A$, see Figure~\ref{fig:overview}~d). Thus $A$ corresponds to the segment endpoints. The vertices $v_1$ and $v_2$ can be chosen arbitrarily.
Now Theorem \ref{thm:main} grants us the existence of a tree $T_Q$ in $Q$. We define $T_0$ to be the graph which we attain when we connect the endpoints of the line segments whenever the corresponding points in $A$ where connected by an edge in $T_Q$. To see that this is possible consider any $2$ segment endpoints $a_1, a_2$ and the corresponding points in $A$ let us denote them by $b_1$ and $b_2$ respectively(not necessarily of the same segment). Observe that if $Q$ and $P_n$ are close enough that $a_1$ sees $a_2$ iff $b_1$ sees $b_2$ (Remember that we assume general position).
See Figure~\ref{fig:overview}~e). 
At last we add the line segments successively to $T_0$. The tree $T_{i}$ is defined to be $T_{i-1}$ whenever the segment $s_i$ is already an edge of $T_i$. Otherwise $T_{i}$ is attained from $T_{i-1}$ by adding segment $s_i$ and removing one of the adjacent edges, which closes a cycle with $s_i$. We choose to remove the edge which is incident to the vertex with the highest degree. Figure~\ref{fig:overview}~f) shows $T_n$.
It is clear that $T_{n}$ is an encompassing PSLG which is a tree as a graph and contains all the given line segments. It remains to show that it has max degree $3$. 
Let us consider a line segment $s_i$. Clearly, if $s_i$ already belonged to $T_{i-1}$ the endpoints have max degree $3$. Otherwise we note that one of the endpoints $u$ has max degree $2$ and the other $v$ has max degree $3$ in $T_{i-1}$. When we add $s_i$ both $u$ and $v$ have max degree $3$ and $4$ respectively and we close some cycle, which contains an edge incident to $u$ and another incident to $v$. If $v$ had indeed degree $4$ we would have remove that edge and $v$ has degree $3$. Thus both have max degree $3$ in $T_i$ and henceforth max degree $3$ in $T_n$.
\end{proof}

\paragraph{Acknowledgments}
The author wants to thank Andrei Asinowski and G\"{u}nter Rote for their eagerness to listen and discuss important and non-important issues~\smiley. Special thanks goes to Matthias Henze for proofreading.

\bibliography{TillLib}
\bibliographystyle{plain}

\end{document}